\documentclass{ws-ijcga}
\usepackage[utf8]{inputenc}
\usepackage{graphicx}
\usepackage{cite}
\usepackage{verbatim}
\usepackage{amsmath,amssymb}
\usepackage{cases}
\usepackage{color}
\usepackage{wrapfig}
\usepackage{subcaption}
\graphicspath{ {./pictures/} }

\newcommand\given{\;\Big\lvert\;}

\DeclareMathOperator*{\extr}{extr}

\title{Fr\'echet similarity of closed polygonal curves}

\author{Schlesinger M.I., Vodolazskiy E.V., Yakovenko V.M.}

\begin{document}

\markboth{Schlesinger M.I., Vodolazskiy E.V., Yakovenko V.M.}{Fr\'echet similarity of closed polygonal curves}

\catchline

\author{ Schlesinger M.I.}

\address{
schles@irtc.org.ua
}

\author{Vodolazskiy E.V.}

\address{
    waterlaz@gmail.com
}

\author{Yakovenko V.M.}

\address{
    asacynloki@gmail.com
}

\address{
International Research and Training Centre \\
of Information Technologies and Systems \\
National Academy of Science of Ukraine \\
Cybernetica Centre, \\
prospect Academica Glushkova, 40, \\
03680, Kiev-680, GSP, Ukraine. 
}

\maketitle

\pub{Received (received date)}{Revised (revised date)}
{Communicated by (Name)}

\begin{abstract}
The article analyzes similarity of closed polygonal curves with respect to the Fr\'echet metric, 
which is stronger than the well-known Hausdorff metric and therefore is more appropriate in some applications.
An algorithm is described that determines whether the Fr\'echet distance between two closed polygonal curves with $m$ and $n$ vertices is less than a given number~$\varepsilon$.
The algorithm takes $O(m n)$ time whereas the previously known algorithms take $O(m n\log (m n))$ time.  

\end{abstract}

%%%%%%%%%%%%%%%%%%%%%%%%%%%%%%%%%%%%%%%%%%%%%%%%%%%%%%%%%%%%%%%%%%%%%
%%%%%%%%                        SECTION                      %%%%%%%%
%%%%%%%%%%%%%%%%%%%%%%%%%%%%%%%%%%%%%%%%%%%%%%%%%%%%%%%%%%%%%%%%%%%%%

\section{Introduction}
The Fr\'echet metric is used for 
cyclic process analysis and image processing \cite{ICDAR.2007.121}.
It is stronger than the well-known Hausdorff metric \cite{Chen02},\cite{rockafellar2011variational},
\cite{ismm2011} and therefore is more appropriate in some applications \cite{frechet}. 
The Fr\'echet metric for closed polygonal curves has been studied in a paper \cite{frechet} by Alt and Godau.
They propose an algorithm that determines whether the distance between two closed polygonal curves with $m$ and $n$ vertices is greater than a given number~$\varepsilon$.
The complexity of the algorithm is $O(m n \log (m n))$ on a random access machine that performs arithmetical operations 
and computes square roots in constant time. 
Our paper shows that the computational complexity of the problem is less than $O(m n \log (m n))$  and provides an algorithm that takes $O(m n)$ time to solve the problem.
The exact formulation of the problem is given in Section \ref{Definitions}.
Sections \ref{Diagrams} describes the concepts of the original paper \cite{frechet}, which with slight modifications serve as the basis for our paper.
The difference between the proposed and known approaches is specified at the end of Section \ref{Diagrams}.
Sections \ref{Achievability}-\ref{ResultSection} describe the proposed approach.

%%%%%%%%%%%%%%%%%%%%%%%%%%%%%%%%%%%%%%%%%%%%%%%%%%%%%%%%%%%%%%%%%%%%%
%%%%%%%%                        SECTION                      %%%%%%%%
%%%%%%%%%%%%%%%%%%%%%%%%%%%%%%%%%%%%%%%%%%%%%%%%%%%%%%%%%%%%%%%%%%%%%

\section{\label{Definitions} Problem definition.}
Let $\mathbb{R}^k$ be a linear space with the Euclidean distance $\text{   }d:\mathbb{R}^k\times \mathbb{R}^k \rightarrow \mathbb{R}$.

\begin{definition}
A closed $m$-gonal curve $X$ is a pair $\langle \bar{x}, f_X \rangle$ where $\bar{x}$ is a sequence 
$(x_0, x_1, \cdots, x_m=x_0)$, $x_i \in \mathbb{R}^k$, and $f_X$ is a function
$[0,m]\rightarrow\mathbb{R}^k$  such that\\
%$f_X:\{ t \in \mathbb{R}\given0 \le t \le m \} \rightarrow \mathbb{R}^k$ such that 
$f_X(i+\alpha)=(1-\alpha) x_i + \alpha x_{i+1}$
for $i \in \{0, 1, \dots , m-1\}$ and $0 \le \alpha \le 1$.
\end{definition}
\begin{definition}
A cyclic shift of an interval $[0,m]$ by a value $\tau \in [0,m]$ is a function
$s:[0,m] \rightarrow [0,m]$ that depends on a parameter $\tau$ such that 
$s(t;\tau) =t+\tau$ for $t+\tau \le m$ and $s(t;\tau) =t+\tau - m$ for $t+\tau > m$. 
\end{definition}
For any number $m$ let $W_m$ be the set of all monotonically non-decreasing continuous functions
$[0,1]\rightarrow[0,m]$ such that $w(0)=0,w(1)=m$. 
\begin{definition}
A function $\varphi:[0,1]\rightarrow \mathbb{R}^k$ is called a monotone reparametrization of a closed $m$-gonal curve $X=\langle \bar{x}, f_X \rangle$
if a function $w \in W_m$ and a number $\tau \in [0,m]$ exist such that $\varphi(t)=f_X(s(w(t);\tau))$ for all $t \in [0,1]$.
\end{definition}
For given closed polygonal curves $X$ and $Y$ denote $\Phi_X$ and $\Phi_Y$ sets of their reparametrizations.
\begin{definition}
The Fr\'echet distance between closed polygonal curves $X$ and $Y$ is
\begin{equation} \nonumber
\delta(X,Y) = \min_{\varphi_X \in \Phi_X} \min_ {\varphi_Y \in \Phi_Y } \max_{0\le t \le 1}d(\varphi_X(t) , \varphi_Y(t)).
\end{equation}
\end{definition}
The problem consists in developing an algorithm that determines whether $\delta(X,Y) \le \varepsilon$
for given closed polygonal curves $X$ and $Y$ and a number~$\varepsilon$.

%%%%%%%%%%%%%%%%%%%%%%%%%%%%%%%%%%%%%%%%%%%%%%%%%%%%%%%%%%%%%%%%%%%%%
%%%%%%%%                        SECTION                      %%%%%%%%
%%%%%%%%%%%%%%%%%%%%%%%%%%%%%%%%%%%%%%%%%%%%%%%%%%%%%%%%%%%%%%%%%%%%%

\section{\label{Diagrams}The free space diagram and pointers.} 
The problem's analysis is based on the concept of a free space diagram introduced by Alt and Godau \cite{frechet} in the following way. 
For two numbers $m$ and $n$ let us define a rectangle $\widetilde{D} = [0,m] \times [0,n]$ with points $(u,v) \in \widetilde{D} $.
For two closed polygonal curves $X$ and $Y$ with $m$ and $n$ vertices and a number $\varepsilon$ a subset
$\widetilde{D}_\varepsilon= \{(u,v) \in  \widetilde{D}\given d(f_X(u),f_Y(v)) \le \varepsilon\}$ is defined. 
Let us also define a rectangle
$D = \widetilde{D} \cup \{(u+m,v)\given(u,v) \in \widetilde{D}\}$ with its subset 
$D_\varepsilon = \widetilde{D}_\varepsilon \cup \{(u+m,v)\given(u,v) \in \widetilde{D}_\varepsilon\}$ called a free space.
Denote $T$, $B$, $L$ and $R$
the top, bottom, left and right sides of the rectangle $D$. Denote $D_{ij}$ a subset $[i-1,i] \times [j-1,j]$ and call it a cell.
\begin{definition} \label {Monotony}
A monotone non-decreasing path (or simply, a monotone path) is a connected subset $\gamma \subset D_\varepsilon$ such that $(u~-~u')(v~-~v') \ge 0$ for any two points 
$(u,v) \in \gamma$, $(u',v') \in \gamma$.
\end{definition}
Note that this definition allows a monotone path to contain vertical segments. 
That is why a condition $(u~-~u')(v~-~v') \ge 0$ is used instead of standard form $(v~-~v')/(u~-~u') \ge 0$ of the definition of a non-decreasing function $v = f(u)$.  
\begin{definition}
Two points $(u,v) \in D$ and $(u',v') \in D$ are mutually reachable if and only if a monotone path $\gamma$ exists such that $(u,v) \in \gamma$, $(u',v') \in \gamma$.
\end{definition}
\begin{definition}
A point $(u,v) \in D_\varepsilon$ is reachable from the bottom if it is reachable from at least one point from $B$;
a point $(u,v) \in D_\varepsilon$ is reachable from the top if it is reachable from at least one point of $T$.
\end{definition}
Denote $g_\downarrow \subset D_\varepsilon$ a set of points reachable from the bottom and
$g^\uparrow \subset D_\varepsilon$ a set of points reachable from the top.

Let us define two pointer functions
\phantom{}\quad $r^\uparrow:g^\uparrow \rightarrow [0,2m] $ and \quad $r_\downarrow:g_\downarrow \rightarrow [0,2m]$. 
For $(u,v) \in g^\uparrow$ the pointer $r^\uparrow(u,v)$ is the maximum value $u^*$ such that $(u,v)$ is reachable from $(u^*,n) \in T$. 
For $(u,v) \in g_\downarrow \setminus B$ the pointer $r_\downarrow(u,v)$ is the maximum value $u^*$ such that $(u,v)$ is reachable from $(u^*,0) \in B$. 
For $(u,0) \in g_\downarrow \cap B$ the pointer $r_\downarrow(u,0)$ equals $u$.  
Figure \ref{Fig1} illustrates the introduced concepts and designations. \\
%Two closed polygonal curves $X$ and $Y$ are shown in Figure \ref{Fig0} and the corresponding rectangles
%$\widetilde{D}$ and  $D$ are shown on Figure \ref{Fig1}, the light region of $D$ being the free space $D_\varepsilon$. %\\
\begin{figure}[h]
  \centering
  \includegraphics*{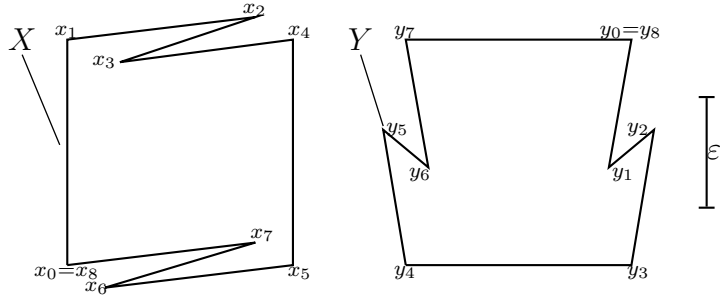}
  \caption{Two closed polygonal curves with an interval of length $\varepsilon$}
  \label{Fig0}
\end{figure}
\begin{figure}[h]
  \centering
  \includegraphics*[width=\textwidth]{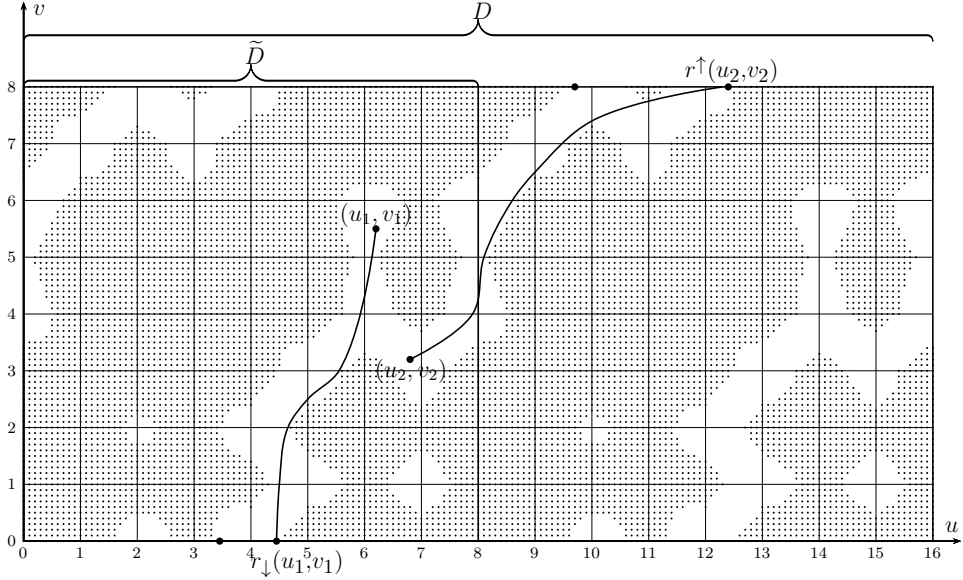}
  \caption{The free space diagram. The light area is the set $D_\varepsilon$.\\
        $r^\uparrow(u_2, v_2)$ is the rightmost reachable point on $T$ from $(u_2, v_2)$.\\
        $r_\downarrow(u_1, v_1)$ is the rightmost reachable point on $B$ from $(u_1, v_1)$.}
  \label{Fig1}
\end{figure}
%The subset $g_\downarrow \subset D_\varepsilon$ is shown by a light area on Figure \ref{Fig2}. 
%The same Figure \ref{Fig2} shows the pointer $r_\downarrow(u,v)$) for some point $(u,v) \in g_\downarrow$. 
%Figure \ref{Fig3} shows the subset $g^\uparrow \subset D_\varepsilon$ and pointer $r^\uparrow(u,v)$.

We rely on the following lemma proved in the paper \cite{frechet} by Alt and Godau (see lemma 9 \cite{frechet}).
\begin{lemma} \label{AltFirst}
The distance between closed polygonal curves $X$ and $Y$  is not greater than $\varepsilon$ if and only if there exists a number 
$u \in [0,m]$, such that
the points $(u + m, n)$ and $(u,0)$ are mutually reachable.
\end{lemma} 
The following lemma is similar to Lemma 10 \cite{frechet} as well as its proof.
\begin{lemma} \label{OurSecond}
Two points $(u_{t},n) \in T$ and $(u_{b},0) \in B$ are mutually reachable if and only if\\ 
\phantom{} \quad\quad $(u_{t},n) \in g_\downarrow $, \quad $(u_{b},0) \in g^\uparrow$, \quad $u_{t} \le r^\uparrow(u_{b},0)$, 
\quad $u_{b} \le r_\downarrow(u_{t},n)$. 
\end{lemma} 
\begin{proof}
Obviously, if $(u_{t},n) \in T$ and $(u_{b},0) \in B$ are mutually reachable then 
$(u_{t},n) \in g_\downarrow $, $(u_{b},0) \in g^\uparrow$ and $u_{t} \le r^\uparrow(u_{b},0)$, $u_{b} \le r_\downarrow(u_{t},n)$.

The reverse implication is also valid, which is illustrated by Figure \ref{Fig5}.
\begin{figure}
  \centering
  \includegraphics{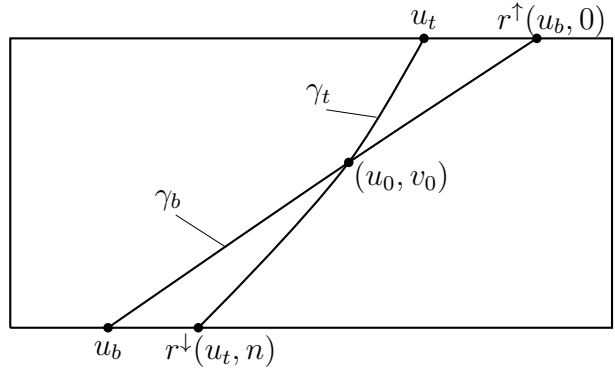}
  \caption{There is a path between $(u_{b}, 0)$ and $(u_{t},n)$}
  \label{Fig5}
\end{figure}
Let $\gamma_{t} \subset D_\varepsilon$ be a monotone path
from $(u_t,n)$ to $(r_\downarrow(u_{t},n),0)$ and  $\gamma_{b} \subset D_\varepsilon$ be a monotone path from $(u_{b},0)$ to $(r^\uparrow(u_{b},0),n)$. 
Both paths are connected subsets that due to conditions $u_{t} \le r^\uparrow(u_{b},0)$, $u_{b} \le r_\downarrow(u_{t},n)$ intersect in at least one point $(u_0, v_0)$.
Let us consider a path $\gamma$ that consists of a segment of $\gamma_{b}$ from $(u_{b},0)$ to $(u_0, v_0)$ 
and a segment of $\gamma^{t}$ from $(u_0, v_0)$ to $(u_{t},n)$. 
The path $\gamma$ is monotone, it is contained inside $D_\varepsilon$ and connects $(u_{b},0)$ and $(u_{t},n)$.
\end{proof}
According to Lemmas \ref{AltFirst} and \ref{OurSecond} testing condition $\delta(X,Y) \le \varepsilon$ is reduced to finding value $u$ that fulfills
\begin{equation} \label{OurDecision}
(u,0) \in g^\uparrow, \quad (u+m,n) \in g_\downarrow, \quad u+m \le r^\uparrow(u,0),
\quad u \le r_\downarrow(u+m,n).  
\end{equation}

The pointer functions $r^\uparrow$ and $r_\downarrow$ are similar to the pointers defined by Alt and Godau.
However, the pointer function $r^\uparrow$ takes values from $T$ and function $r_\downarrow$ takes values from $B$,
while Alt and Godau consider pointers with values from $T \cup R$ and with values from $L \cup B$.
The Alt's and Godau's pointers allow to use a divide and conquer type algorithm that
merges either vertically or horizontally two diagrams with known pointers and
obtains a bigger diagram with pointers for the new diagram.
The pointers can be computed in $O(1)$ time for diagrams containing only one cell. 
By sequentially merging smaller diagrams into bigger ones, the pointers for the whole 
diagram can be obtained in $O(m n \log (m n))$ time.

In this paper we rely on a recurrent relation between pointer values $r^\uparrow$ and $r_\downarrow$ on cell borders.
It is trivial to compute $r^\uparrow$ on $T$ and $r_\downarrow$ on $B$.
Our algorithm does not use the divide and conquer approach but proceeds cell by cell. 
It propagates pointers on each cell's borders using the recurrent relation
and eventually obtains the values $r^\uparrow$ on $B$ and $r_\downarrow$ on $T$ in $O(mn)$ time.

\section{\label{Achievability}Formal properties of pointer functions.}
For each $1 \le i \le 2m$ and $1 \le j \le n$ denote  $T_{ij}$ and $R_{ij}$ the top and right borders of a square cell $D_{ij}=[i-1,i] \times [j-1,j]$ 
and extend these denotations so that $T_{i0}$ is a bottom border of $D_{i1}$ and $R_{0j}$ is  a left border of $D_{1j}$. 
Let us designate 
\begin{eqnarray*}
    & B_{ij} = T_{i(j-1)},\quad & L_{ij} = R_{(i-1)j},\\
    & TR_{ij} = T_{ij} \cup R_{ij}, \quad
    & LB_{ij} = L_{ij} \cup B_{ij}.
\end{eqnarray*}

In order to test (\ref{OurDecision}) the sets $g^\uparrow$, $g_\downarrow$ and functions $r^\uparrow$, $r_\downarrow$
have to be expressed with a finite data structure.
It is significant that intersection 
$D_\varepsilon \cap D_{ij}$ is convex for each pair $(i,j)$ \cite{frechet}. 
It is not difficult to prove that for each pair $(i,j)$ 
the intersections $g_\downarrow \cap T_{ij}$, $g_\downarrow \cap R_{ij}$ are also convex, 
so that each of these intersections is a single interval (line segment) on the border of $D_{ij}$. 
Simple recurrent relations hold for these intervals, 
so that intervals $g_\downarrow \cap T_{ij}$ and $g_\downarrow \cap R_{ij}$ can be computed based on $g_\downarrow \cap B_{ij}$ and $g_\downarrow \cap L_{ij}$ 
as well as $g^\uparrow \cap B_{ij}$ and $g^\uparrow \cap L_{ij}$ can be computed based on $g^\uparrow \cap T_{ij}$ and $g^\uparrow \cap R_{ij}$. 
The Algorithm 1 described in \cite{frechet} for slightly different purposes is an example of this computation. 
Each step of the recursion takes constant time, consequently, 
computing intersections $g_\downarrow \cap T_{ij}$, $g_\downarrow \cap R_{ij}$, $g^\uparrow \cap B_{ij}$ and $g^\uparrow \cap L_{ij}$ for all $(i,j)$ takes $O(mn)$ time. 
Therefore, from now on it is assumed that these line segments are available.

As for the functions $r^\uparrow$ and $r_\downarrow$ they require more detailed analysis.
Let us define partial ordering $\preccurlyeq$ on the set $D$ such that 
$(u_1,v_1) \preccurlyeq (u_2,v_2)$ if and only if $u_1 \le u_2$ and $v_1 \ge v_2$.
On each subset $TR_{ij}$ and $LB_{ij}$ relation $\preccurlyeq$ is a complete ordering.
Therefore, for each closed subset $b \subset TR_{ij}$ and for each closed subset $b \subset LB_{ij}$
a symbol $\extr\limits_\preccurlyeq b$ will be used as a designation of such point $(u^*,v^*) \in b$ 
that $(u',v') \preccurlyeq (u^*,v^*)$ for each $(u',v') \in b$.   
Both $r_\downarrow$ and $r^\uparrow$ are monotone functions of their argument in a sense of the following lemma.  
\begin{lemma} \label{MyMonotony}\\
Let $(u_1, v_1), (u_2, v_2) \in g_\downarrow$. If $(u_1,v_1) \preccurlyeq (u_2,v_2)$ then $r_\downarrow(u_1, v_1) \le r_\downarrow(u_2, v_2)$.\\
Let $(u_1, v_1), (u_2, v_2) \in g^\uparrow$. If $(u_1,v_1) \preccurlyeq (u_2,v_2)$
then $r^\uparrow(u_1, v_1) \le r^\uparrow(u_2, v_2)$.
\end{lemma} 
\begin{proof}
Let $\gamma_1$ and $\gamma_2$ be two monotone paths that connect  $(u_1, v_1)$ with $(r_\downarrow(u_1, v_1), 0)$ and $(u_2, v_2)$ with $(r_\downarrow(u_2, v_2),0)$ respectively.
Let us assume that $r_\downarrow(u_1, v_1) > r_\downarrow(u_2, v_2)$. 
As Figure~\ref{Fig6} shows the paths $\gamma_1$ and $\gamma_2$ intersect at some point  $(u_0,v_0)$.
Therefore, a monotone path that connects $(u_2,v_2)$ with $(r_\downarrow(u_1, v_1), 0)$ exists.
This path consists of a segment of $\gamma_1$ from $(r_\downarrow(u_1, v_1), 0)$ to $(u_0,v_0)$ and a segment of $\gamma_2$ from $(u_0,v_0)$ to  $(u_2, v_2)$. 
This means that the point $(u_2,v_2)$ is reachable from a point that is located to the right of the point $(r_\downarrow(u_2, v_2),0)$.
This contradicts with the definition of function $r_\downarrow$. 
Therefore, the assumption $r_\downarrow(u_1, v_1) > r_\downarrow(u_2, v_2)$ is proved to be wrong and the first statement of the theorem is proved. 
The proof of the second statement is similar.
\begin{figure}
  \centering
  \includegraphics{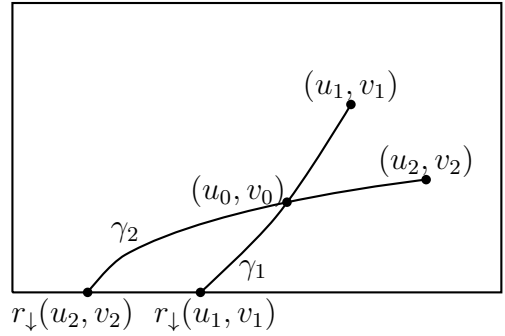}
  \caption{Monotonicity of $r_{\downarrow}$}
  \label{Fig6}
\end{figure}     
\end{proof}
Due to convexity of $D_\varepsilon \cap D_{ij}$ and monotonicity of functions $r_\downarrow$ and $r^\uparrow$ 
they satisfy the following recursive relations that will allow to use some sort of dynamic programming for their computation.
For $(u,v) \in g_\downarrow \cap TR_{ij}$ it holds that
\begin {equation} \label{RecGenDown}
r_\downarrow(u,v) = r_\downarrow\left(\extr\limits_\preccurlyeq\left\{(u'v') \in g_\downarrow \cap LB_{ij} \given  u' \le u, v' \le v \right\}\right)
\end {equation}
and for $(u,v) \in g^\uparrow \cap LB_{ij}$
\begin {equation} \label{(RecGenUp)}
r^\uparrow(u,v) = r^\uparrow\left(\extr\limits_\preccurlyeq\left\{(u'v') \in g^\uparrow \cap TR_{ij}\given u' \ge u, v' \ge v \right\}\right).
\end {equation}
Relations (\ref{RecGenDown}) and (\ref{(RecGenUp)}) immediately result in the following lemma.
\begin{lemma} \label{Constancy}
For each pair $i \in \{1, \dots , 2m\}$, $j \in \{1, \dots , n \}$ 
the pointer $r_\downarrow$ is constant on $g_\downarrow \cap R_{ij}$ and
the pointer $r^\uparrow$ is constant on $g^\uparrow \cap B_{ij}$. 
\end{lemma}
\begin{proof} 
For each pair of points $(u,v) \in R_{ij}$ and $(u',v') \in LB_{ij}$ the condition $u' \le u$ in (\ref{RecGenDown}) is fulfilled and may be omitted. 
For each point $(i,v) \in g_\downarrow \cap R_{ij}$ a point $(u',v') \in g_\downarrow \cap LB_{ij}$ exists such that $v' \le v $. 
Consequently, this condition in (\ref{RecGenDown}) also may be omitted. 
Relation (\ref{RecGenDown}) becomes
$$r_\downarrow(u,v) = r_\downarrow\Big(\extr\limits_\preccurlyeq\big\{(u',v') \in g_\downarrow \cap LB_{ij}\big\}\Big),$$ 
and $r_\downarrow(u,v)$ does not depend on $u$ and $v$, which proves the first statement of the lemma.\\
For each pair of points $(u,v) \in B_{ij}$ and $(u',v') \in TR_{ij}$ the condition $v' \ge v$ in (\ref{(RecGenUp)}) is fulfilled and may be omitted. 
For each point $(i,v) \in g^\uparrow \cap B_{ij}$ a point $(u',v') \in g^\uparrow \cap TR_{ij}$ exists such that $u' \ge u $. 
Consequently, this condition in (\ref{(RecGenUp)}) also may be omitted. 
Relation (\ref{(RecGenUp)}) becomes 
$$r^\uparrow(u,v) = r^\uparrow\Big(\extr\limits_\preccurlyeq\big\{(u',v') \in g^\uparrow \cap TR_{ij}\big\}\Big),$$ 
which proves the second statement of the lemma.
\end{proof}
Now relation (\ref{RecGenDown}) can be written in more detail. 
\begin{figure}[h]
 \centering
 \begin{subfigure}[b]{0.47\textwidth}
        \includegraphics[width=\textwidth]{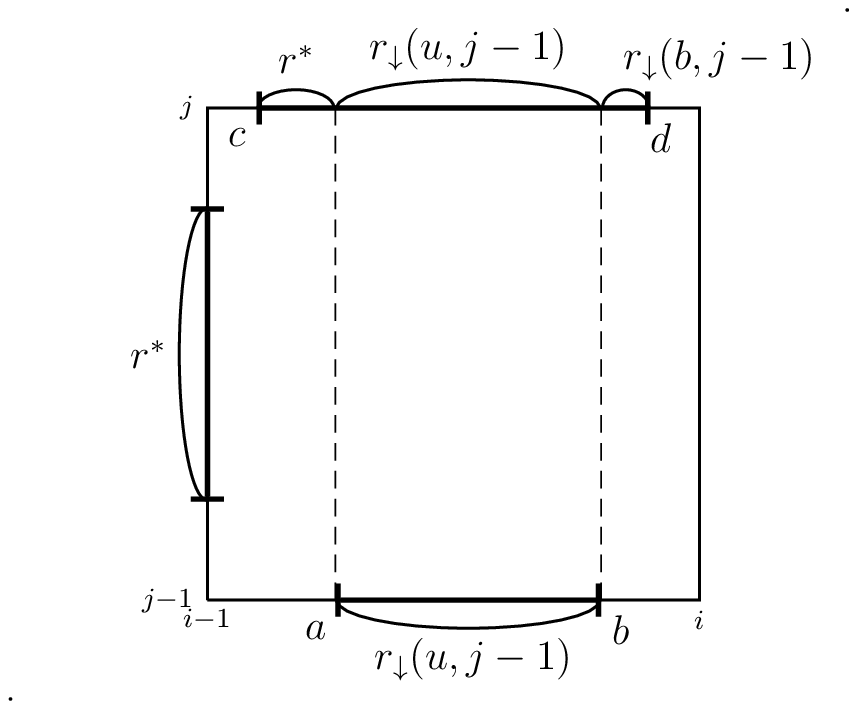}
        \caption{Pointer function $r_\downarrow$ on the upper border depends on $r_\downarrow$ on the lower and left border of the cell.}
        \label{FigRecursive_TB}
  \end{subfigure}
 \hspace{0.03\textwidth}
 \begin{subfigure}[b]{0.47\textwidth}
        \includegraphics[width=\textwidth]{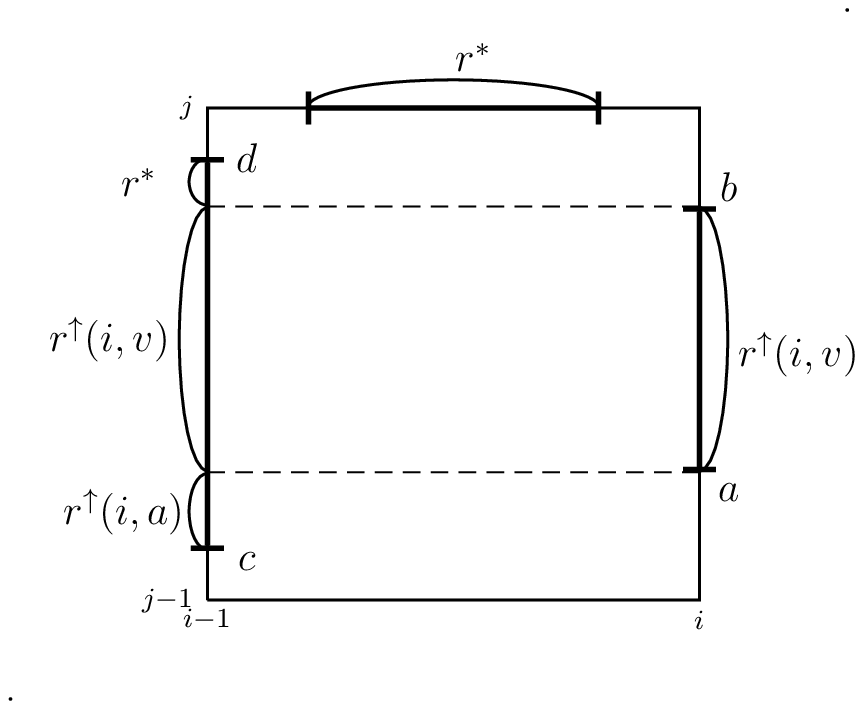}
        \caption{Pointer function $r^\uparrow$ on the left border depends on $r^\uparrow$ on the upper and right border of the cell.}
        \label{FigRecursive_LR}
  \end{subfigure} 
  \caption{Recursive dependency of pointer functions.}
\end{figure}
Let $r^*$ be the constant value of the pointer function $r_\downarrow$ 
on $g_\downarrow \cap L_{ij}$ for some $(i,j)$ and
$$ [a, b] = \left\{u\given(u,j-1) \in g_\downarrow \cap B_{ij}  \right\}$$ as it is shown on Fig \ref{FigRecursive_TB}.\\
If $g_\downarrow \cap B_{ij} = \emptyset$ then for $(u,v) \in g_\downarrow \cap TR_{ij}$
\begin{equation} \label{Recursion1}
    r_\downarrow (u,v)= r^*.
\end{equation}
If $g_\downarrow \cap B_{ij} \ne \emptyset$ then for $(u,j) \in g_\downarrow \cap T_{ij}$ and for $(i,v) \in g_\downarrow \cap R_{ij}$
\begin{numcases}{r_\downarrow (u,j) =}
    r^*                     &  $\mbox{if } u < a$, \label{Recursion2}\\
    r_\downarrow (u,j-1)    &  $\mbox{if } a \le u < b$, \label{Recursion3}\\
    r_\downarrow (b,j-1)    &  $\mbox{if } u \ge b$, \label{Recursion4}
\end{numcases}
\begin{equation} \label{Recursion5}
r_\downarrow (i,v)= r_\downarrow (b,j-1). \quad\quad\quad\quad\quad\quad\quad\quad
\end{equation}
Similarly, relation (\ref{(RecGenUp)}) can be specified. 
Let $r^*$ be the constant value of the pointer function $r^\uparrow$ on $g^\uparrow \cap T_{ij}$ for some $(i,j)$ and
$$[a, b] = \left\{v\given(i,v) \in g^\uparrow \cap R_{ij}\right\}$$
 as it is shown on Fig \ref{FigRecursive_LR}.\\ 
If $g^\uparrow \cap R_{ij} = \emptyset$ then for $(u,v) \in g^\uparrow \cap LB_{ij}$
 \begin{equation} \label{Recursion6}
 r^\uparrow (u,v)= r^*.
 \end{equation}
If $g^\uparrow \cap R_{ij} \ne \emptyset$ then for $(i-1,v) \in g^\uparrow \cap L_{ij}$ and for $(u,j-1) \in g^\uparrow \cap B_{ij}$
%\begin{equation}\label{Recursion7_9}
\begin{numcases}{r^\uparrow (i-1,v)= }
    r^*                 & $\mbox{if } v > b$, \label{Recursion7}\\
    r^\uparrow (i,v)    & $\mbox{if } a < v \le b$, \label{Recursion8}\\
    r^\uparrow (i,a)    & $\mbox{if } v \le a$, \label{Recursion9}
\end{numcases}
\begin{equation} \label{Recursion10}
r^\uparrow (u,j-1)= r^\uparrow (i,a). \quad\quad\quad\quad\quad\quad\quad\quad
\end{equation}
The following Lemma \ref{QuantityOfIntSimilar} states that restriction of $r^\uparrow$ to $g^\uparrow \cap R_{ij}$ is a piecewise constant function. 
Lemma \ref{QuantityOfInt} states that the restriction of $r_\downarrow$ to  $g_\downarrow \cap T_{ij}$ is also piecewise constant with the exception of at most one interval, where it is linear.

For any set $S$ we say that $I$ is a partition of $S$ if $S=\bigcup\limits_{int \in I} int$ and $int \cap int' = \emptyset$ for each pair $int,int' \in I$, $int \ne int'$. 
\begin{lemma}\label{QuantityOfIntSimilar}
For each $(i,j)$ a partition $I^\uparrow(i,j)$ of the set $g^\uparrow \cap R_{ij}$ to intervals exists 
such that the function $r^\uparrow$ is constant on each interval in $I^\uparrow(i,j)$; moreover, $|I^\uparrow(i,j)| \le 4m + 1$.
\end{lemma}
\begin{proof}
Indeed, either $g^\uparrow \cap R_{(2m)j}$ is empty or pointer $r^\uparrow$ is constant on $g^\uparrow \cap R_{(2m)j}$ and equals $2m$.
Therefore, $I^\uparrow(2m,j)$ consists of no more than one interval.
It follows from recursive relations (\ref{Recursion6})-(\ref{Recursion10}) that
there are no more than two intervals that belong to $I^\uparrow(i-1,j)$ and do not belong to $I^\uparrow(i,j)$.
The first interval comes from (\ref{Recursion7}).
The second interval comes from (\ref{Recursion9}).
The pointer $r^\uparrow (i-1,v)$ is constant on each of these two intervals.
Therefore,  $|I^\uparrow(2m,j)| \le 1$, $|I^\uparrow(i-1,j)| \le |I(i,j)|+2$,
and finally $|I^\uparrow(i,j)| \le 4m -2i +1 \le 4m +1$. 
\end{proof}
\begin{lemma}\label{QuantityOfInt}
For each $(i,j)$ a partition $I_\downarrow(i,j)$ of $g_\downarrow \cap T_{ij}$ into intervals exists with the following properties: \\
-- there is no more than one interval $int \in I(i,j)$ such that $r_\downarrow (u,j) = u$ on $int$;\\
-- function $r_\downarrow$ is constant on all other intervals;\\
-- moreover, $|I_\downarrow(i,j)| \le 2n +1$.
\end{lemma}
\begin{proof}
By definition, $r_\downarrow (u,0)=u$ for each $(u,0) \in g_\downarrow \cap T_{i0}$.
Therefore, partition $I_\downarrow(i,0)$ consists of no more than one interval.
It follows from recursive relations (\ref{Recursion1})-(\ref{Recursion5}) that
there are no more than two intervals that belong to $I_\downarrow(i,j)$ and do not belong to $I_\downarrow(i,j-1)$.
The first of them is included according to relation (\ref{Recursion2}).
The second interval appears in partition $I_\downarrow(i,j)$ when $u \ge b$ in relation (\ref{Recursion4}).
The function $r_\downarrow (u,j)$ is constant on each of these two intervals.
Therefore,  $|I_\downarrow(i,0)| \le 1$, $|I_\downarrow(i,j)| \le |I_\downarrow(i,j-1)|+2$, and finally $|I_\downarrow(i,j)| \le 2j +1 \le 2n+1$.
\end{proof}
According to Lemma \ref{Constancy} the set $g^\uparrow \cap B$ and the restriction of a function $r^\uparrow$ to this set can be expressed with subsets $g^\uparrow \cap B_{i1}$, $i \in \{1,2, \dots , 2m \}$, 
and values $r_i^\uparrow$ of a function $r^\uparrow$ on these subsets.
According to Lemma \ref{QuantityOfInt} the set $g_\downarrow \cap T$ and the restriction of $r_\downarrow$ to this set can be expressed with
the sets $I_\downarrow(i,n)$ of intervals $int$ and with numbers $r_\downarrow^{int}$, where $int \in I_\downarrow(i,n)$, $i \in \{1,2, \dots , 2m \}$.
Numbers $r_\downarrow^{int}$ determine the function $r_\downarrow$ on $int$ so that if $r_\downarrow^{int}$ is less than the right endpoint of $int$ then $r_\downarrow(u,n)=r_\downarrow^{int}$ for all $(u,n) \in int$.  
Otherwise, $r_\downarrow(u,n)=u$ for all $(u,n) \in int$. 
\begin{lemma}\label{TestingOfMainInequality}
%Given the sets $g^\uparrow \cap T_{i0}$ with pointer values $r_i^\uparrow$ for $i \in \{1,2, \dots , m \}$
%and partitions $I_\downarrow(i+m,n)$ with pointer values $r_\downarrow^{int}$ for $i \in \{1,2, \dots , m \}, int \in I_\downarrow(i+m,n)$,
%testing $\delta(X,Y) \le \varepsilon$ takes $O(m n)$ time.
Let $X$ and $Y$ be closed polygonal curves and for each $i \in \{1,2, \dots , m \}$ the following data be known:\\
- the set $g^\uparrow \cap T_{i0}$ with pointer $r_i^\uparrow$;\\
- partition $I_\downarrow(i+m,n)$ of $g_\downarrow \cap T_{(i+m)n}$;\\
- value $r_\downarrow^{int}$ for each interval $int \in I_\downarrow(i+m,n)$;\\
then testing $\delta(X,Y) \le \varepsilon$ takes $O(m n)$ time.

\end{lemma}
\begin{proof} According to Lemmas \ref{AltFirst} and \ref{OurSecond}, inequality $\delta(X,Y) \le \varepsilon$ is equivalent to the existence of a number $u \in [0,m]$ 
that fulfills (\ref{OurDecision}).
Such $u$ exists if and only if a triple 
$i \in \{ 1,2, \dots , m \}$, $int \in I_\downarrow(i+m,n)$, $u \in [0,m]$ exists that fulfills conditions
\begin{equation}\label{OurDecisionConcreteNew }
(u,0) \in g^\uparrow \cap T_{i0}, \quad (u+m,n) \in int, \quad u+m \le r_i^\uparrow,
\quad u \le r_\downarrow(u+m,n).  
\end{equation}
The condition $u \le r_\downarrow(u+m,n)$ in (\ref{OurDecisionConcreteNew }) can be replaced with condition $u \le r_\downarrow^{int}$
independently of whether the pointer $r_\downarrow(u,n)$ takes constant value $r_\downarrow^{int}$ on $int$ or $r_\downarrow(u,n)=u$.
In both cases, condition (\ref{OurDecisionConcreteNew }) is equivalent to 
\begin{equation} \nonumber %\label{OurDecisionAuxiliary}
(u,0) \in g^\uparrow \cap T_{i0}, \quad (u+m,n) \in int, \quad u+m \le r_i^\uparrow,
\quad u \le r_\downarrow^{int},
\end{equation}
or in more detail
\begin{equation} \label{OurDecisionWeekConcrete2 }
c_{i} \le u \le d_{i}, \quad a^{int}-m \le u \le b^{int}-m, \quad u \le r_{i}^\uparrow - m,
\quad u \le r_\downarrow^{int},
\end{equation}
where 
 $a^{int}$ and $b^{int}$ are the left and right endpoints of the interval $int$ and $c_{i}$ and  $d_{i}$ are the horizontal coordinates of the leftmost and the rightmost points of $g^\uparrow \cap T_{i0}$. 
In turn, a triple $i \in \{ 1,2, \dots , m \}$, $int \in I_\downarrow(i+m,n)$, $u \in [0,m]$ that satisfies (\ref{OurDecisionWeekConcrete2 }) exists iff a pair 
$i \in \{ 1,2, \dots , m \}$, $int \in I_\downarrow(i+m,n)$ exists  that satisfies
$$\max \{c_i,  a^{int}-m  \} \le 
\min \{ d_i,   b^{int}-m, r_i^\uparrow - m, r_\downarrow^{int} \}.$$ 
Testing the last inequality takes constant time for any pair $(i,int)$. 
According to Lemma \ref{QuantityOfInt} the number of intervals tested for each $i$ does not exceed $(2n+1)$. Therefore, the number of tested pairs $(i,int)$ is $O(m n)$. 
\end{proof}

%%%%%%%%%%%%%%%%%%%%%%%%%%%%%%%%%%%%%%%%%%%%%%%%%%%%%%%%%%%%%%%%%%%%%
%%%%%%%%                        SECTION                      %%%%%%%%
%%%%%%%%%%%%%%%%%%%%%%%%%%%%%%%%%%%%%%%%%%%%%%%%%%%%%%%%%%%%%%%%%%%%%

\section{Computing pointer functions}\label{InputDatar}
\subsection{The general scheme}
The partition of $T$ and $B$ into intervals that represent $r^\uparrow$ and $r_\downarrow$ is obtained 
by two similar independent algorithms called the forward and backwards pass. 
Both passes consist of $2mn$ steps and in each step compute a pointer function on a border of some cell.
From now on we will refer to the pair $(i, j)$ as the number of the step.

The forward pass proceeds cell by cell from left to right and from bottom to top and computes pointers $r_\downarrow$
on $TR_{ij}$ based on pointers $r_\downarrow$ on $LB_{ij}$ according to (\ref{Recursion1})-(\ref{Recursion4}) starting from cell $(1, 1)$.
The result of the forward pass are the partitions of $T_{(i+m)n}$,  $1 \le i \le m$, that represent $r_\downarrow$ on $T$.

Similarly, the backward pass starts from cell $(2m, n)$ and proceeds from right to left and from top to bottom and computes
$r^\uparrow$ on $LB_{ij}$ based on $r^\uparrow$ on $TR_{ij}$ according to (\ref{Recursion6})-(\ref{Recursion9}).
The result of the backward pass are the constant values $r_i^\uparrow$ of the pointer function $r^\uparrow$ on $g^\uparrow \cap T_{i0}$ for $1 \le i \le m$.

During the forward pass the function $r_\downarrow$ on $g_\downarrow \cap T_{ij}$ as the partition $I_\downarrow(i,j)$ of $T_{ij}$ is stored in the following data structure.
The function $r_\downarrow$ on each interval $int \in I_\downarrow(i,j)$ is represented by a triple $(beg^{int}, end^{int}, r_\downarrow^{int})$,
where $beg^{int}$ and $end^{int}$ are the endpoints of interval $int$ and $r_\downarrow^{int}$ has the meaning defined right before Lemma \ref{TestingOfMainInequality}.
The triples $(beg, end, val)$ sorted by their endpoints are stored in a double-ended queue (deque) with the following operations performing in constant time:\\
-- reading and removing either the leftmost or the rightmost triple;\\
-- pushing a triple either to the left or to the right end of the deque.\\
 For a given number $x$ the operations of cutting the deque to the left of $x$ and cutting to the right of $x$ are defined.
Cutting to the left of $x$ means removing all triples $(beg, end, val)$ where $end < x$ from the deque 
and replacing a triple $(beg, end, val)$ where $beg \le x < end$ with a triple $(x, end, val)$.
Cutting to the right of $x$ means removing all triples $(beg, end, val)$ where $beg > x$ and
replacing a triple $(beg, end, val)$ where $beg \le x \le end$ either with a triple $(beg, x, val)$ if $val < end$, or with a triple $(beg, x, x)$ if $val = end$. 
Since the triples in the deque are sorted, the time spent on cutting the deque is proportional to the number of triples removed from the deque. 
Note that each cut of the deque performs only one push to the deque.

%In order to analyze computational complexity of algorithms a pair reading-removing is considered as an indivisible operation. 
%So, reading without removing includes an insertion because is implemented as a reading-removing with subsequent insertion.

The partitions $I^\uparrow(i,j)$ of the sets $g^\uparrow \cap R_{ij} $ and functions $r^\uparrow$ on $g^\uparrow \cap R_{ij} $ are 
stored in the same data structures. Of course, instead of the leftmost and the rightmost triple there are the lowest and the highest triple. 

The forward and the backward passes rely on the sets 
$g_\downarrow \cap R_{ij}$, $g_\downarrow \cap T_{ij}$, $g_\uparrow \cap R_{ij}$ and $g_\uparrow \cap T_{ij}$ to be precomputed.
It has been mentioned at the beginning of Section \ref{Achievability} that these sets can be computed in a straightforward way.

\subsection{The forward pass}\label{Forward}
The forward pass works with $2m$ deques $Q_\downarrow(i)$,  $i \in \{1, 2, \dots , 2m \}$, 
whose content depends on the number $(i,j)$ of the step. 
The $(i,j)$-th step starts with the known value $r^*(i-1,j)$ of $r_\downarrow$ on $L_{ij}$ and with $Q_\downarrow(i)$ representing $r_\downarrow$ on $B_{ij}$.
The $(i,j)$-th step updates the deque $Q_\downarrow(i)$ to represent $r_\downarrow$ on $T_{ij}$ and computes the constant value $r^*(i,j)$ of $r_\downarrow$ on $R_{ij}$.

As it has been shown during the proof of Lemma \ref{QuantityOfInt} a partition of non-empty set $g_\downarrow \cap B_{i1}$ consists of a single interval. 
Therefore, at the start of the $(i,1)$-th step the deque $Q_\downarrow(i)$ is 
either empty (if $g_\downarrow \cap B_{i1} = \emptyset$) or includes a single triple $(beg, end, end)$ 
where $beg$ and $end$ are the horizontal coordinates of the leftmost and the rightmost points in $g_\downarrow \cap B_{i1}$.
The pointers  $r^*(0,j)$ are obviously equal to $0$ for each non-empty $g_\downarrow \cap L_{1j}$.
  
When all four sets $g_\downarrow \cap R_{ij}$,  $g_\downarrow \cap T_{ij}$, $g_\downarrow \cap L_{ij}$,  $g_\downarrow \cap B_{ij}$ are non-empty 
the update of $Q_\downarrow(i)$ and the computation of $r^*(i,j)$ is done as follows.
Let
$$[a, b] = \left\{u\given(u,j-1) \in g_\downarrow \cap B_{ij}  \right\}, \quad [c, d] = \left\{u\given(u,j) \in g_\downarrow \cap T_{ij}    \right\},$$
1. {\bf if} $c<a$ {\bf then} push $(c, a, r^*(i-1,j))$ to the left end of $Q_\downarrow(i)$;\\
2. \phantom{{\bf if} $c<a$} {\bf else} cut $Q_\downarrow(i)$ to the left of $c$;\\  
3. read a triple $(beg, end, r^*)$ from the right and save the value $r^*$; \\
4. {\bf if} $b<d$ {\bf then} push $(b, d, r^*)$ to the right end of $Q_\downarrow(i)$;\\
5. \phantom{{\bf if} $b<d$} {\bf else} cut $Q_\downarrow(i)$ to right of $d$;\\
6. $r^*(i,j)=r^*$.\\
The Operations 1, 4 and 6 represent relations (\ref{Recursion2}), (\ref{Recursion4}) and (\ref{Recursion5}) directly. 
The relation (\ref{Recursion3}) is represented with Operations 2 and 4 indirectly in a sense that the uncut part of $Q_\downarrow(i)$ remains unchanged.

It is not necessary to consider all special cases when some of the sets 
$g_\downarrow \cap R_{ij}$,  $g_\downarrow \cap T_{ij}$, $g_\downarrow \cap L_{ij}$,  $g_\downarrow \cap B_{ij}$ are empty. 
In each of these cases the $(i,j)$-th step consists of some part of Operations 1-6 or their slight modifications. 
Since we are mainly interested in the computational complexity of the algorithm, considering the above described complete case is sufficient.

One can see that one step of the forward pass does not perform in constant time due to the cutting of the deque that can take $O(j)$ time for the $(i,j)$-th step. 
Nevertheless, the complexity of the forward pass is $O(m n)$ as the following lemma states.  

\begin{lemma}\label{ComplexityForward}
It takes $O(m n)$ time to complete the forward pass.
\end{lemma}
\begin{proof}
The forward pass starts with initializing the deques $Q_\downarrow(i)$, $i \in \{1,2, \dots, 2m\}$, and the pointers $r^*(0,j) = 0$, $j \in \{1, 2, \dots, n\}$. 
Evidently, it takes $O(m + n)$ time. 

No more than three triples are pushed in deque on each step. 
One triple is pushed either with Operation 1 or with  Operation 2, the second is pushed either with Operation 4 or with Operation 5. 
The third push is made with Operation 3 because the triple that was read and removed from the deque has to be returned to the deque. 
Therefore, no more than $6mn$ triples are pushed during all steps.

Reading and removing triples are fulfilled with Operations 2, 3 and 4. 
Number of these operations in $(i,j)$-th step may differ from number of insertions in this step. 
However, the number of readings and removing during the whole forward pass cannot exceed the total number of insertions, consequently, cannot be greater than $6mn$.
\end{proof}              
\subsection{The backward pass}\label{Backward}
The backward pass works with $n$ deques $Q^\uparrow(j)$,  $j \in \{1, 2, \dots , n \}$, whose content depends on the number $(i,j)$ of the step. 
The $(i,j)$-th step starts with the known value $r^*(i,j)$ of the pointer $r^\uparrow$ on $g^\uparrow \cap T_{ij}$
and with the deque $Q^\uparrow(j)$ representing $r^\uparrow$ on $g^\uparrow \cap R_{ij}$.
The $(i,j)$-th step updates $Q^\uparrow(j)$ to represent $r^\uparrow$ on $g^\uparrow \cap L_{ij}$ 
and computes the value $r^*(i,j-1)$ of the pointer $r^\uparrow$ on $g^\uparrow \cap B_{ij}$.

When all four sets $g^\uparrow \cap R_{ij}$,  $g^\uparrow \cap T_{ij}$, $g^\uparrow \cap L_{ij}$,  $g^\uparrow \cap B{ij}$ are non-empty the update of $Q^\uparrow(j)$ is done as follows.
Let
$$
[a, b] = \left\{v\given(i,v) \in g^\uparrow \cap R_{ij}  \right\}, \quad [c, d] = \left\{v\given(i-1,v) \in g^\uparrow \cap L_{ij} \right\},
$$
1. {\bf if} $b < d$ {\bf then} push $(b, d, r^*(i,j))$ to the upper end of $Q^\uparrow(j)$;\\
2. \phantom{{\bf if} $b < d$} {\bf else} $Q^\uparrow(j)$ is cut to the up of $d$;\\  
3. read a triple $(beg, end, r^*)$ from the lower end and save the value $r^*$; \\
4. {\bf if} $c \le a$ {\bf then} push $(c, a, r^*)$ to the lower end of $Q^\uparrow(j)$;;\\
5. \phantom{{\bf if} $c \le a$} {\bf else}  cut $Q^\uparrow(j)$ down of $c$;\\
%according to (\ref{Recursion5})\\
6. $r^*(i,j-1)=r^*$.\\
The Operations 1, 4 and 6 represent relations (\ref{Recursion7}), (\ref{Recursion9}) and (\ref{Recursion10}) directly. 
The relation (\ref{Recursion8}) is represented with Operations 2 and 5 indirectly in a sense that the uncut part of $Q^\uparrow(j)$ remains unchanged.

\begin{lemma}\label{ComplexityBackward}
It takes $O(m n)$ time to complete the backward pass.
\end{lemma}
\begin{proof}
The proof is based on the same considerations as the proof of Lemma \ref{ComplexityForward}. 
\end{proof}

%%%%%%%%%%%%%%%%%%%%%%%%%%%%%%%%%%%%%%%%%%%%%%%%%%%%%%%%%%%%%%%%%%%%%
%%%%%%%%                        SECTION                      %%%%%%%%
%%%%%%%%%%%%%%%%%%%%%%%%%%%%%%%%%%%%%%%%%%%%%%%%%%%%%%%%%%%%%%%%%%%%%

\section{The result}\label{ResultSection}
\begin{theorem} 
Let $X$ and $Y$ be closed polygonal curves with $m$ and $n$ vertices and $\delta(X,Y)$ be the Fr\'echet distance between them. 
Testing the inequality $\delta(X,Y) \le \varepsilon$ takes $O(m n)$ time.
\end{theorem}
\begin{proof}
Testing the inequality $\delta(X,Y) \le \varepsilon$ is reduced to the following computations.
\\
Computing the sets $D_\varepsilon \cap T_{ij}$ and $D_\varepsilon \cap R_{ij}$, $0 \le i \le 2m$, $0 \le j \le n$, 
is reduced to solving $2mn$ quadratic equations and takes $O(mn)$ time.
\\
Computing subsets $g_\downarrow \cap T_{in}$ and $g^\uparrow \cap B_{i1}$, $ 1 \le  i \le 2m $, takes $O(m n)$ time.\\  
Due to Lemma \ref{ComplexityForward} computing the restrictions of $r_\downarrow$ to $g_\downarrow \cap T_{in}$  takes $O(m n)$ time.
\\Due to Lemma \ref{ComplexityBackward} computing the restrictions of  $r^\uparrow$ to $g^\uparrow \cap B_{i1}$  takes $O(m n)$ time.
\\
Due to Lemma \ref{TestingOfMainInequality} testing the inequality $\delta(X,Y) \le \varepsilon$ based on these data takes $O(m n)$ time.
\end{proof}

\end{document}